\documentclass[final,3p,12pt,authoryear,lefttitle]{elsarticle}
\usepackage[T1]{fontenc}
\usepackage{newtxtext,newtxmath}
\usepackage[english]{babel}
\usepackage{CJKutf8}
\usepackage{mathtools} 
\usepackage{natbib}
\usepackage{booktabs}
\usepackage{threeparttable}
\usepackage{graphicx}
\usepackage{multirow}
\usepackage{supertabular}
\usepackage[hyperfootnotes]{hyperref}
\usepackage[hyperref]{xcolor}
\hypersetup{colorlinks=true,pdfauthor=author}
\usepackage[font=normalsize,labelfont=bf]{caption}
\usepackage{colortbl}
\usepackage{float}
\usepackage{longtable}
\usepackage{pdflscape}

\newtheorem{proposition}{Proposition}
\newtheorem{remark}{Remark}

\newenvironment{proof}{\paragraph{Proof}}{\hfill$\blacksquare$}
\makeatletter
\def\ps@pprintTitle{%
\let\@oddhead\@empty
\let\@evenhead\@empty
\def\@oddfoot{\footnotesize\itshape
\ifx\@journal\@empty Elsevier
\else\@journal\fi\hfill\today}%
\let\@evenfoot\@oddfoot}
\makeatother
\journal{Preprint}%
\begin{document}
%
\begin{frontmatter}

\title{\textbf{Artificial intelligence and the skill premium}\tnoteref{t1}}


%
\tnotetext[t1]{Corresponding author: David E. Bloom}
\author[harvard]{David E. Bloom}
\ead{dbloom@hsph.harvard.edu}
\author[wu]{Klaus Prettner}
\ead{klaus.prettner@wu.ac.at}
\author[unistra]{Jamel Saadaoui}
\ead{saadaoui@unistra.fr}
\author[qdl]{Mario Veruete}
\ead{veruete.mario@quantum-datalab.com}
\address[harvard]{Department of Global Health and Population, Harvard TH Chan School of Public Health, Boston, MA, 02115, USA}
\address[wu]{Vienna University of Economics and Business, Department of Economics, Welthandelsplatz 1, 1020 Vienna, Austria}
\address[unistra]{University of Strasbourg, University of Lorraine, BETA, CNRS, 67000, Strasbourg, France}
\address[qdl]{Quantum DataLab, 33000, Bordeaux, France }

\vspace{150mm}
\begin{abstract}
\noindent What will likely be the effect of the emergence of ChatGPT and other forms of artificial intelligence (AI) on the skill premium? To address this question, we develop a nested constant elasticity of substitution production function that distinguishes between industrial robots and AI. Industrial robots predominantly substitute for low-skill workers, whereas AI mainly helps to perform the tasks of high-skill workers. We show that AI reduces the skill premium as long as it is more substitutable for high-skill workers than low-skill workers are for high-skill workers.
\begin{keyword} Automation, Artificial Intelligence, ChatGPT, Skill Premium, Wages, Productivity.
\JEL J30, O14, O15, O33.
\end{keyword}
\end{abstract}
\end{frontmatter}

\newpage

\section{Introduction}

\label{sec:intro}

In recent decades, industrial robots have become an increasingly important substitute for workers performing relatively routine mechanical tasks in the manufacturing sector. The worldwide stock of operative industrial robots has increased strongly, particularly since the global economic and financial crisis of 2008-2009 \citep[cf.][]{Abeliansky2020, PrettnerBloom2020_adapted, Jurkat2022}. Recent research indicates that this trend has put downward pressure on the wages of low-skill workers, much more so than high-skill workers \citep[cf.][]{AcemogluRestrepo2018a, AcemogluRestrepo2019, 
dauth_adjustment_2021, Cords2022}. As a consequence, the skill premium has increased \citep[cf.][]{Lankisch2019, PrettnerStrulik2019published}.

With the emergence of ChatGPT in the fall of 2022 and, more generally, with the impressive improvements made in artificial intelligence (AI) recently, the question arises as to how the future evolution of the skill premium will be affected \citep[cf.][]{AcemogluRestrepo2018lowskill}. This is because, in contrast to industrial robots, AI predominantly substitutes for tasks performed by high-skill workers. For example, AI-based models and devices are increasingly used to diagnose diseases, develop drugs, write reports, code, or simply generate inspiring ideas in fields such as marketing and research and development. Since these tasks are often non-routine and performed by high-skill workers, AI may put downward pressure on their wages and thereby also on the skill premium.

To analyze the effects of AI on the skill premium at the aggregate level, we develop a general nested constant elasticity of substitution (CES) production function in which robots substitute for low-skill workers and AI substitutes for high-skill workers. We allow for imperfect substitutability of workers with different skill levels by robots and AI and derive a condition under which the emergence of AI would reduce the skill premium.

\section{AI and the skill premium: theoretical considerations}
\label{sec:mod}

Since automation in terms of industrial robots predominantly affects low-skill workers performing routine mechanical tasks, whereas automation in terms of AI predominantly affects high-skill workers, we develop a nested CES production function for the analysis of the differential effects of industrial robots and AI on wages. Consider that the aggregate production function is given by
\begin{equation}\label{eq:aggprod2}
Y_t := {K_t}^\alpha{\left[\beta_3 {\left(\beta_1 {L_{u,t}}^\theta+(1-\beta_1){P_t}^\theta\right)}^{\tfrac{
\gamma}{\theta}}+(1-\beta_3){\left(\beta_2 {L_{s,t}}^\varphi +(1-\beta_2) {G_t}^{\varphi} \right)}^{\tfrac{\gamma}{\varphi}}\right]}^{\tfrac{1-\alpha}{\gamma}}. 
\end{equation}
Here, $Y_t$ is output in period $t$; $\beta_1$, $\beta_2$, and $\beta_3$ refer to input shares; $L_{u,t}$ is employment of low-skill workers; $P_t$ denotes the stock of operative industrial robots; $\theta$ determines the elasticity of substitution between low-skill workers and robots in low-skill intensive production tasks; $L_{s,t}$ denotes employment of high-skill workers; $G_t$ refers to the stock of high-skill replacing AI; $\varphi$ determines the elasticity of substitution between high-skill workers and AI in high-skill intensive production tasks; $\gamma$ determines the elasticity of substitution between low-skill intensive and high-skill intensive production tasks; $K_t$ refers to the traditional physical capital stock (machines, assembly lines); and $\alpha$ is the elasticity of output with respect to traditional capital input. For the attainable values of the parameters, we consider the reasonable range $\alpha,\beta_1,\beta_2,\beta_3,\in(0,1)$ and $\gamma,\theta,\varphi \in  (0,1]$.\footnote{The case of $\theta=\varphi=1$ is analyzed by \cite{Hufnagl2023}.}  



Using this production function\footnote{\cite{Steigum2011}, \cite{Lankisch2019}, \cite{Prettner2017}, \cite{AntonyKlarl2019}, and \cite{Cords2022} use various functions that are nested in our general CES production function to analyze the effects of automation on economic growth, wages, the labor income share, and unemployment, but all without the presence of AI.}, assuming perfect competition, and normalizing the price of final output to unity, allows us to derive the wage rates of low-skill and high-skill workers ($w_u$ and $w_s$, respectively) as the marginal product of the corresponding production factor: 
\begin{align}
    w_{u} & := \frac{\partial Y_t}{\partial L_{u,t}}\\
          & = (1-\alpha)\beta_1\beta_3 {K_t}^\alpha {L_{u,t}}^{\theta-1}{\left(\beta_1 {L_{u,t}}^\theta+(1-\beta_1){P_t}^\theta\right)}^{\tfrac{\gamma}{\theta}-1}\nonumber\\
          &\times {\left[\beta_3 {\left(\beta_1 {L_{u,t}}^\theta+(1-\beta_1){P_t}^\theta\right)}^{\tfrac{\gamma}{\theta}}+(1-\beta_3){\left(\beta_2 {L_{s,t}}^\varphi+(1-\beta_2){G_t}^\varphi \right)}^{\tfrac{\gamma}{\varphi}}\right]}^{\tfrac{1-\alpha-\gamma}{\gamma}},\nonumber
\end{align}
and
\begin{align}
    w_{s} & :=\frac{\partial Y_t}{\partial L_{s,t}}\\
    & = (1-\alpha) \beta_2 (1-\beta_3) {K_t}^\alpha {\left( \beta_2 {L_{s,t}}^{\varphi} +(1-\beta_2) {G_t}^\varphi\right)}^{\tfrac{\gamma}{\varphi}-1} \nonumber\\
&\times {\left[ \beta_3 {\left(\beta_1 {L_{u,t}}^\theta +(1-\beta_1) {P_t}^\theta\right)}^{\tfrac{\gamma}{\theta}} +(1-\beta_3) {\left(\beta_2 {L_{s,t}}^\varphi+(1-\beta_2){G_t}^\varphi \right)}^{\tfrac{\gamma}{\varphi}}\right]}^{\tfrac{1-\alpha-\gamma}{\gamma}}. \nonumber
\end{align}
Dividing $w_s$ by $w_u$ yields the skill premium, i.e., the factor by which the wages of high-skill workers exceed the wages of low-skill workers: 
\begin{align}
\frac{w_s}{w_u}    & = \frac{\beta_2(1-\beta_3)}{\beta_1\beta_3} {L_{s,t}}^{\varphi-1} {L_{u,t}}^{1-\theta} {\left(\beta_1 {L_{u,t}}^\theta +(1-\beta_1) {P_t}^\theta\right)}^{1-\tfrac{\gamma}{\theta}}  \times {\left(\beta_2 {L_{s,t}^\varphi+(1-\beta_2){G_t}^\varphi}\right)}^{\tfrac{\gamma}{\varphi}-1}.\nonumber
\end{align}


The skill premium is an important measure of wage inequality. Its explicit expression allows us to establish the following central result.

\begin{proposition}
The growing use of AI, ceteris paribus, reduces wage inequality between high-skill and low-skill workers, as long as AI is more substitutable for high-skill workers than low-skill workers are for high-skill workers.
\end{proposition}

\begin{proof}

To show this, we compute the derivative of the skill premium with respect to the use of AI as
\begin{align}\label{eq:wageratio_G}
    \frac{\partial (w_s/w_u)}{\partial G_t} & = \frac{\varphi \beta_2(1-\beta_2)(1-\beta_3)\left(\frac{\gamma}{\varphi}-1\right)}{\beta_1\beta_3}{G_t}^{\varphi-1} {L_{s,t}}^{\varphi-1} {L_{u,t}}^{1-\theta}\\
    & \times{\left(\beta_1{L_{u,t}}^{\theta} + (1-\beta_1) {P_t}^\theta \right)}^{1-\tfrac{\gamma}{\theta}} {\left(\beta_2 {L_{s,t}}^\varphi +(1-\beta_2) {G_t}^\varphi \right)}^{\tfrac{\gamma}{\varphi}-2}. \nonumber
\end{align}
Since we assume that $L_{u,t}>0$, $L_{s,t}>0$, $P_t>0$, and $G_t>0$ hold for any $t\geq0$, it follows that the sign of $\partial_{G_t}(w_s/w_u)$ is determined by the sign of $\gamma/\varphi-1$. In particular, $w_s/w_u$ decreases in $G_t$ if and only if $\gamma<\varphi$.
\end{proof} 

\begin{remark}
The use of AI is neutral in its impact on the skill premium if $\gamma/\varphi=1$, which implies that AI is equally substitutable for high-skill workers as low-skill workers are.
\end{remark}

The intuition behind this result is that the deployment of AI replaces high-skill workers directly, but it also substitutes for low-skill workers to the extent that high-skill intensive production tasks can substitute for low-skill intensive production tasks according to the CES production structure. As long as substitution is easier between AI and high-skill workers than between low-skill and high-skill workers, the deployment of AI has stronger effects on the wages of high-skill workers than on the wages of low-skill workers. Thus, increasing the use of AI reduces the skill premium.

\section{AI and the skill premium: numerical illustration}
\label{sec:numerics}

To illustrate the effects of AI on the skill premium, we rely on the parameter values and initial conditions summarized in Table \ref{table:parameter_values} and simulate the evolution of the skill premium for an increase in $G_t$. We take a conventional value $\alpha=1/3$ for the elasticity of output with respect to physical capital \citep[cf.][]{Jones1995, Acemoglu2009}, set $\gamma=1/3$ so that the elasticity of substitution between low-skill and high-skill intensive tasks lies comfortably in the range of plausible values \citep[cf.][]{Acemoglu2002, Acemoglu2009}, choose $\theta=3/4$ to get an elasticity of substitution between low-skill workers and industrial robots in low-skill intensive production tasks of 4, and set $\varphi=1/2$ so that AI is not as good a substitute for high-skill workers in performing high-skill intensive tasks as industrial robots are for low-skill workers in performing low-skill intensive tasks. The value of $K_t$ is taken from the \cite{FRBL2023} for the year 2019 and the value of $P_t$ is constructed for the same year following \cite{Prettner2023} who relies on data from the \cite{IFR2022} for the number of operative industrial robots and a projection of robot prices based on the data reported by \cite{Jurkat2021, Jurkat2022}. Finally, we take the employment data for $L_u$ and $L_s$ from the \cite{BLS2020}, assuming that high-skill workers are those with a bachelor's degree or higher, while low-skill workers do not have a university degree.

\begin{table}[h]
\footnotesize
\centering
\caption{Summary of Parameter Values and Initial Levels for the Simulation}
\label{table:parameter_values}
\begin{tabular}{lll}
\hline \hline
\textbf{Parameter} & \textbf{Value} & \textbf{Source / Justification} \\ \hline
$K$ & 69.0 Trillion US\$ & \cite{FRBL2023} \\ 
$L_u$ & 98.3 Million persons & \cite{BLS2020} \\ 
$L_s$ & 58.4 Million persons & \cite{BLS2020} \\ 
$P$ & 17.3 Billion US\$ & \cite{IFR2022}; \cite{Jurkat2021, Jurkat2022} \\ 
$\alpha$ & 1/3 & \cite{Acemoglu2009}; \cite{Jones1995} \\ 
$\gamma$ & 1/3 & \cite{Acemoglu2002} \\ 
$\theta$ & 3/4 & \cite{Jurkat2022} \\ 
$\varphi$ & 1/2 & Chosen such that $0<\varphi<\theta\leq1$ \\ 
$\beta_1$ & 0.9 & The central result is robust to changes in this parameter\\
$\beta_2$ & 0.95 & The central result is robust to changes in this parameter \\
$\beta_3$ & 2/3 & The central result is robust to changes in this parameter \\
\hline
\end{tabular}
\end{table}


In Table \ref{table:results}, we show the simulation results for different AI use values as reflected in $G_t$\footnote{We can observe through Equation \eqref{eq:wageratio_G} that Proposition 1 is invariant to the values of $\beta_1,$ $\beta_2$, and $\beta_3$.}. In the first row, we assume that AI is not yet used in the production process such that $G_t=0$. This leads to a skill premium of about 2, i.e., wages of high-skill workers are twice the wages of low-skill workers. Increasing the use of AI reduces the skill premium. In the second row, $G_t$ is half the value of $P_t$ and the skill premium decreases to about 1.7. In the third row, the value of $G_t$ is now the same as the value of $P_t$ so that the skill premium shrinks further to 1.62. Finally, in the last row, we assume that the value of AI has exceeded the value of industrial robots by a factor of two, which causes the skill premium to shrink to 1.52.

\begin{table}[ht]
\footnotesize
\centering
\caption{Skill premium for various levels of AI ($G_t$)}
\label{table:results}
\begin{tabular}{ll}
\hline \hline
$G_t$ & $w_s/w_u$ \\ \hline
$G_t=0$ & 2.0023 \\
$G_t = 0.5 \cdot P_t$ & 1.6979 \\
$G_t=P_t$ & 1.6152 \\
$G_t=2 \cdot P_t$ & 1.5213 \\
 \hline
\end{tabular}
\end{table}





Our numerical illustrations show that, indeed, AI has the potential to reduce the skill premium. However, three cautionary notes are in order. First, the result depends on the difference between $\varphi$ and $\gamma$. If the values of these parameters are close to each other, the skill premium is relatively insensitive to increasing AI. This shows the importance of having reliable estimates of the relevant elasticities of substitution at the aggregate level. Second, $G_t$ would not change in isolation in reality. Other variables such as $P_t$, $K_t$, and the share of high-skill workers can increase at the same time. If $P_t$ increases in addition to $G_t$, some of the dampening effect of $G_t$ on the skill premium is offset. Finally, there could be labor augmenting technological progress, which raises the productivity of low- and high-skill workers. Such changes would obscure the isolated effect of AI on the skill premium in observable data.

\section{Conclusions}
\label{sec:concl}

We explore the effects of the emergence of AI on the skill premium. To this end, we develop a nested CES production function in which industrial robots predominantly substitute for low-skill workers, whereas AI predominantly substitutes for high-skill workers. We show analytically and numerically that AI has the potential to reduce the skill premium and thereby mitigate or even reverse observed increases in inequality that have emerged in recent decades. 

In future research, it would be useful to construct precise estimates for the relevant elasticities of substitution to assess the plausibility of our parameter assumptions. Furthermore, the production structure with robots as low-skill automation and AI as high-skill automation could be introduced into full-fledged general equilibrium models to analyze the effects of AI on economic growth, employment, and welfare \citep[cf.][]{AcemogluRestrepo2018a, PrettnerStrulik2019published, sequeira2021robots, Cords2022, Hemous2022, Shimizu2023, thuemmel2023optimal}. While doing so is beyond the scope of this paper, such a framework would allow consideration of i) the dynamic effects of the endogenous accumulation of the different capital stocks in the model, ii) an endogenous education decision of whether to stay low-skill or to become high-skill, which would allow for richer dynamics of the skill premium, and iii) the 
evolution of social welfare subject to different welfare functions from egalitarian (Rawlsian) to utilitarian (Benthamite or Millian).

\bibliographystyle{apalike}
\bibliography{References_2}

\newpage

\begin{center}
  \Large{\textbf{Appendix --- Artificial intelligence and the skill premium}}
\end{center}

\vspace{0.1cm}

\begin{center}
  {Authors: David E. Bloom\footnote{Corresponding author: David E. Bloom}, Klaus Prettner, Jamel Saadaoui, Mario Veruete}
\end{center}

\vspace{0.1cm}

\section{Notebook for replication and additional results}

All the computations of the paper and dynamic graphical visualization can be retrieved from the following Wolfram notebook:
\begin{center}
\href{https://www.wolframcloud.com/obj/mariov/Published/Appendix.nb}{https://www.wolframcloud.com/obj/mariov/Published/Appendix.nb}
\end{center}
In this appendix, for purposes of clarity, we follow \cite{Lankisch2019} and \cite{PrettnerBloom2020_adapted} and recall the analytical effect of automation with and without AI on the wage levels of high-skill  workers. The aggregate production function allows for two types of labor but does not consider AI\footnote{Note that equation \eqref{eq:prodwithoutAI} is a special case of equation (4) in the main text of the paper when $G_t=0$, $\theta=1$, $\beta_1=1/2$, $\beta_2={\left(\frac{\beta-1}{1-2^\gamma\beta}\right)}^{\varphi/\gamma}$, and $\beta_3=2^\gamma \beta$.}. In this case, we have
\begin{equation}\label{eq:prodwithoutAI}
    Y_t=\left[(1-\beta) L_{s,t}^\gamma+\beta\left(P_t+L_{u,t}\right)^\gamma\right]^{\frac{1-\alpha}{\gamma}} K_t^\alpha,
\end{equation}
where $L_{s,t}$ refers to the number of high-skill workers, $L_{u,t}$ to the number of low-skill workers, $\beta \in(0,1)$ is the weight of low-skill workers in the production process, and $\gamma \in(-\infty, 1]$ determines the elasticity of substitution between both types of workers. For $\gamma=1$, workers with different skills are perfect substitutes and for $\gamma \rightarrow-\infty$, they are perfect complements.

Without AI, the effect of an increase in the stock of automation capital on the wage of high-skill workers is given by:
$$
\frac{\partial w_s}{\partial P_t}=(1-\alpha) Y_t \frac{(1-\beta) \beta L_{s,t}^\gamma}{L_{s,t}\left(P_t+L_{u,t}\right)^{1-\gamma}} \frac{1-\alpha-\gamma}{\left[(1-\beta) L_{s,t}^\gamma+\beta\left(P_t+L_{u,t}\right)^\gamma\right]^2}=\left(\begin{array}{l}
\geq 0 \text { for } 1-\alpha \geq \gamma, \\
<0 \text { for } 1-\alpha<\gamma 
\end{array}\right).
$$
The sign of this derivative is ambiguous and depends on the substitutability between both types of workers. If $\gamma$ is high, such that substitution is easy, increasing the stock of robots would reduce the wages of high-skill workers. If, in contrast, the substitutability between low-skill workers and high-skill workers is low, automation will raise the wages of high-skill workers.

The growing use of AI has important consequences for the wages of high-skill workers, as shown in Table 2 in the main text of the paper. Following logic similar to that in the case without AI, we can see that the growing use of AI will decrease the wages of high-skill workers when the substitutability between high-skill workers and AI ($\varphi$) is higher than the substitutability between low-skill workers and high-skill workers ($\gamma$), as shown in Proposition 1 in the main text of the paper. The implications in terms of the skill premium are now different from those of automation in terms of an increasing stock of industrial robots.

\begin{figure}[!ht]
    \begin{center}
    \caption{Skill premium for various levels of AI}
    \label{fig:Figure1}
    \includegraphics[width=0.65\linewidth]{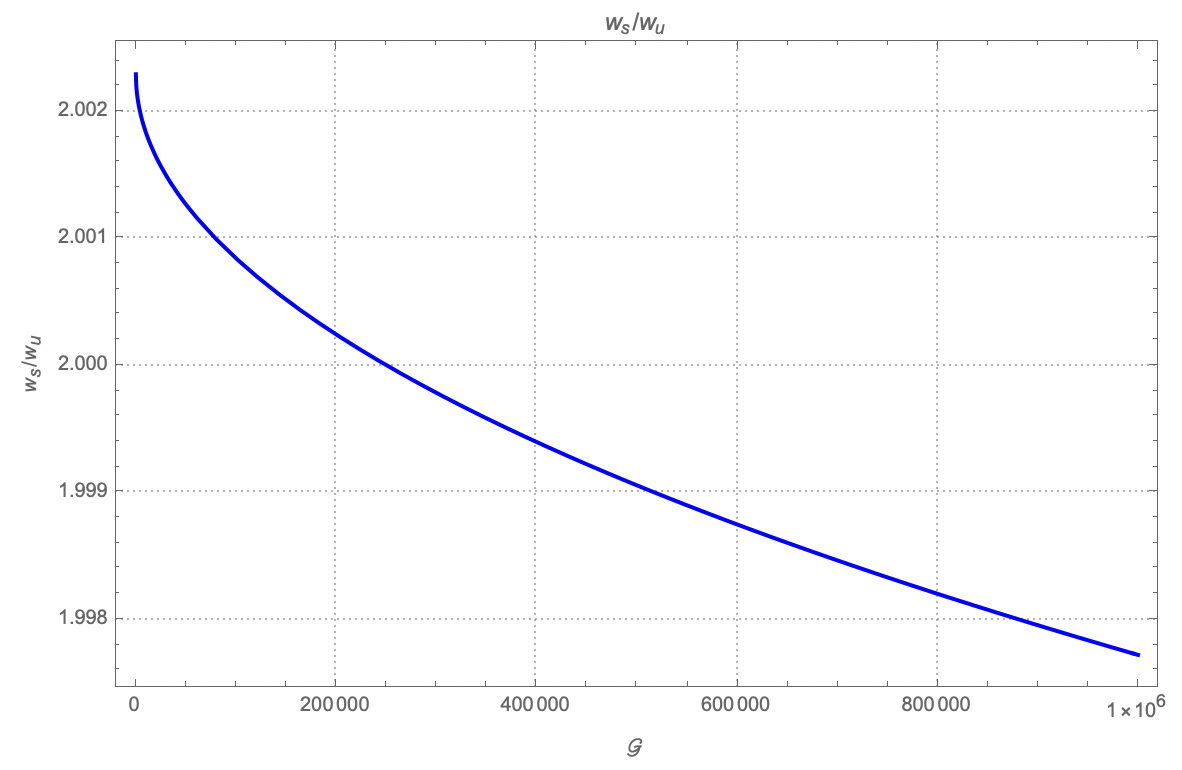}
    
    \footnotesize{Note: The different levels of AI are on the x-axis.}
        \end{center}
    
\end{figure}

We illustrate the complex interplay between AI ($G$) and industrial robots ($P$) for wage outcomes on 3D graphs in Figures \ref{fig:ws}, \ref{fig:wu}, and \ref{fig:wswu}, capturing the surfaces represented by \(w_s(G,P) \), \(w_u(G,P)\), and the ratio \(w_s/w_u(G,P)\). In these graphs, the red points signify high levels of industrial robots, but a complete absence of AI ($G=0$). By contrast, the blue points indicate the presence of both high AI and high industrial robot levels.

For the wages of high-skill workers, as represented by \(w_s\) in Figure \ref{fig:ws}, the graph reveals the following. The wage is at its peak when AI is at a minimum but the stock of industrial robots is high. An increase in AI leads to a noticeable decline in \( w_s \). 

At the other end of the spectrum, the wage rate of low-skill workers, denoted by \( w_u \) in Figure \ref{fig:wu}, offers a different perspective. It tends to be highest in environments with minimal influence from AI and industrial robots. However, the red point, indicative of an environment with a high stock of industrial robots but no AI, presents a challenge for low-skill workers, causing a low wage. Yet, there is a silver lining: When AI and the stock of industrial robots rise in tandem, low-skill workers witness an increase in their wage rate.

Shifting the focus to the wage disparity, represented by the \( w_s/w_u \) ratio in Figure \ref{fig:wswu}, the graph paints a vivid picture of the evolving wage dynamics. The skill premium is most pronounced at the red point, where no AI and a large stock of industrial robots tip the scale in favor of high-skill workers, amplifying their wage considerably more than their low-skill counterparts. But technology also offers a way to narrow this gap. As the blue point suggests, when both AI and the stock of industrial robots reach their zenith, they interact and reduce this wage disparity.


\begin{figure}
    \centering
    \includegraphics[width=10cm]{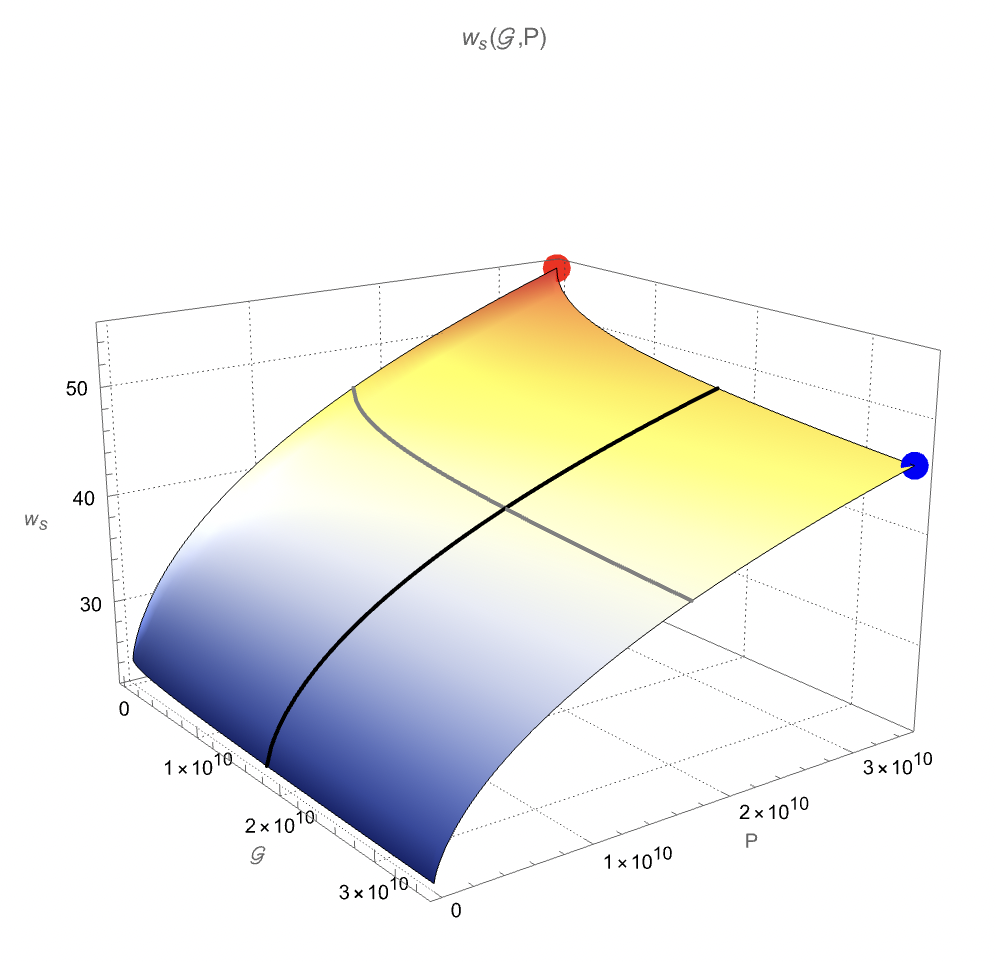}
    \caption{High-skill wages ($w_s$) for various levels of AI and Industrial Robots}
    \label{fig:ws}
\end{figure}

\begin{figure}
    \centering
    \includegraphics[width=10cm]{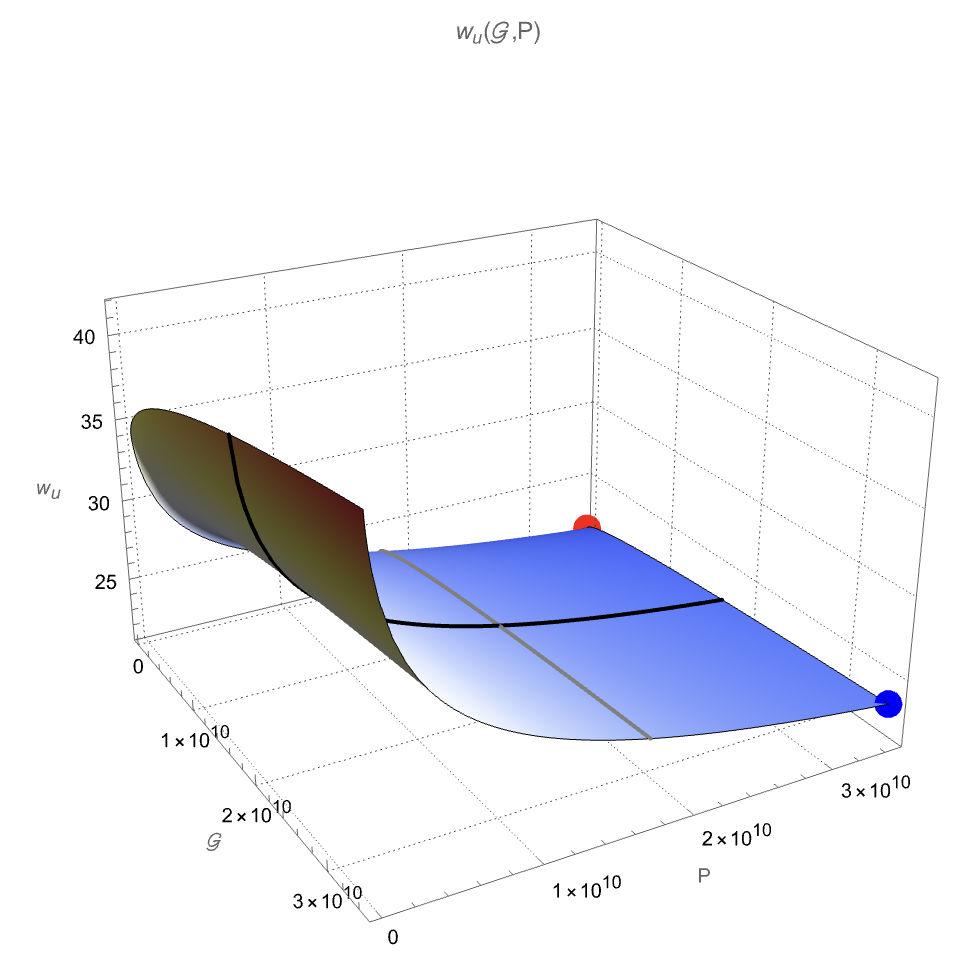}
    \caption{Low-skill wages ($w_u$) for various levels of AI and Industrial Robots}
    \label{fig:wu}
\end{figure}

\begin{figure}
    \centering
    \includegraphics[width=10cm]{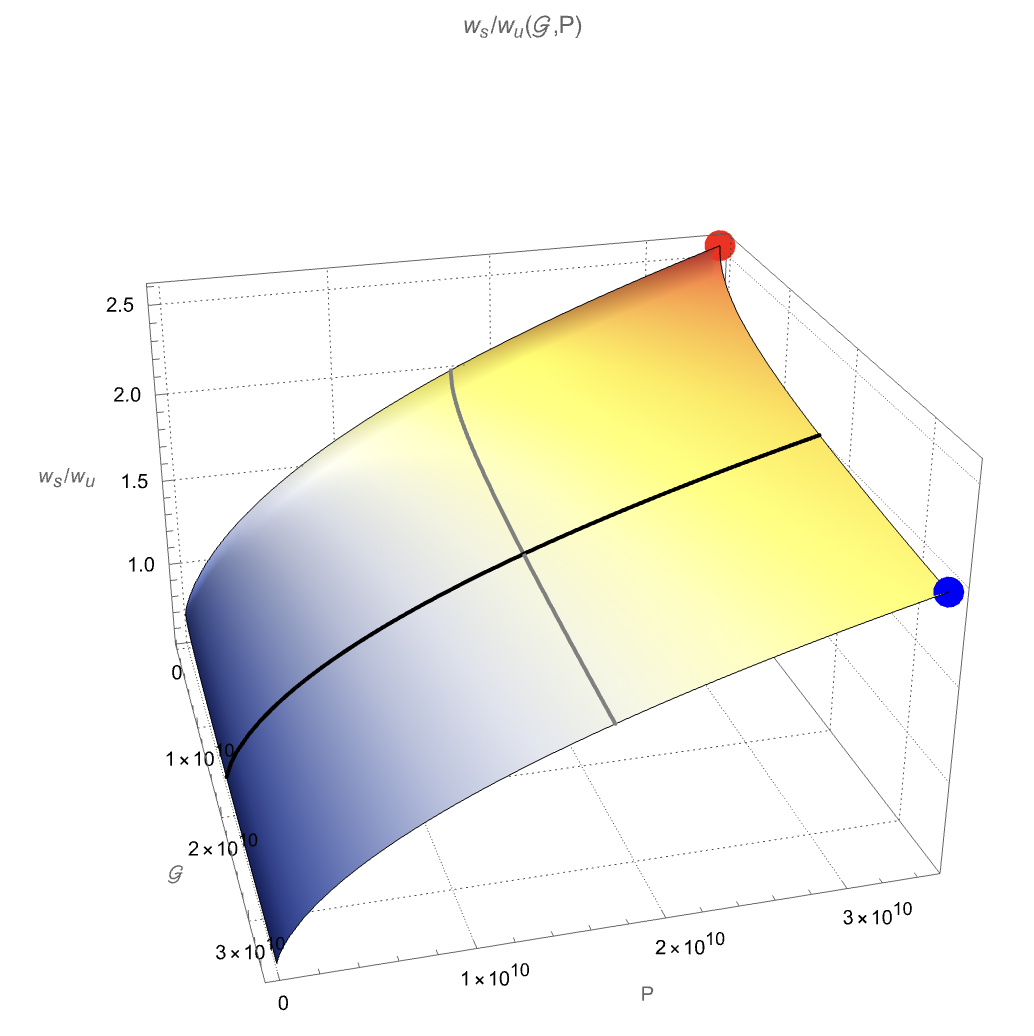}
    \caption{Skill premium for various levels of AI and Industrial Robots}
    \label{fig:wswu}
\end{figure}

\end{document}